\documentclass[3p]{elsarticle}

\usepackage{makeidx}
\usepackage{amsfonts}
\usepackage{amsmath}
\usepackage{amssymb}
\usepackage{algorithmic}
\usepackage{algorithm}
\usepackage[bottom]{footmisc}
\usepackage{enumerate}
\usepackage{hyperref}

\newtheorem{theorem}{Theorem}[section]
\newtheorem{lemma}[theorem]{Lemma}
\newtheorem{claim}[theorem]{Claim}
\newtheorem{proposition}[theorem]{Proposition}
\newtheorem{corollary}[theorem]{Corollary}

\newtheorem{problem}{Problem}
\newproof{proof}{Proof}
\newproof{claim_proof}{Proof of Claim}
\newcommand{\heading}[1]{\medskip\par\noindent{\bf #1}}

\newenvironment{packed_enum}{
	\begin{enumerate}
		\setlength{\itemsep}{1pt}
	    \setlength{\parskip}{0pt}
		\setlength{\parsep}{0pt}
}{\end{enumerate}}

\newenvironment{packed_itemize}{
	\begin{itemize}
		\setlength{\itemsep}{1pt}
	    \setlength{\parskip}{0pt}
		\setlength{\parsep}{0pt}
}{\end{itemize}}

\def\computationproblem#1#2#3{
	\vskip 1ex
	\begin{center}
	\fbox{\begin{tabular}{rp{10cm}}
	{\bf Problem:\enspace}&#1\\
	{\bf Input:\enspace}&#2\\
	{\bf Output:\enspace}&#3\\
	\end{tabular}}
	\end{center}
	\vskip 1ex
}
\def\crossnumber{\mathrm{cr}}
\def\claimqed{\hfill$\diamond$}
\def\eps{\varepsilon}

\def\O{\mathcal{O}{}}
\def\cP{{\sffamily P}}
\def\cNP{{\sffamily NP}}

\def\embed{\hookrightarrow}

\def\clpsat{{\sc Clause-Linked Planar 3-SAT}}
\def\inner{\text{inner}}
\def\outer{\text{outer}}

\def\rest{\upharpoonright}

  \def\calC{{\cal C}}

\def\restcontr#1{{\sc $#1$-RestrictedContract}}
\def\contr{{\sc Contract}}
\def\subcontr{{\sc $\ell$-SubContract}}
\def\matchcontr{{\sc MatchingContract}}

\begin{document}

\title{MSOL Restricted Contractibility to Planar Graphs\tnoteref{support}}
\tnotetext[support]{The conference version of this paper appeared in IPEC 2012~\cite{akkv}.  The
first author acknowledges support of Special focus on Algorithmic Foundations of the Internet, NSF
grant \#CNS-0721113 and mgvis.com {\tt http://mgvis.com}. The second and the third authors are
supported by CE-ITI (P202/12/G061 of GA\v{C}R) and Charles University as GAUK 196213.}

\author[dimacs]{James Abello}
\ead{abello@dimacs.rutgers.edu}
\author[cunicsi]{Pavel Klav\'{\i}k}
\ead{klavik@iuuk.mff.cuni.cz}
\author[cunidam]{Jan Kratochv\'{\i}l}
\ead{honza@kam.mff.cuni.cz}
\author[cunidam,rutgers]{Tom\'{a}\v{s} Vysko\v{c}il}
\ead{whisky@kam.mff.cuni.cz}

\address[dimacs]{DIMACS Center for Discrete Mathematics and Theoretical Computer Science, Rutgers University, Piscataway, NJ}
\address[cunicsi]{Computer Science Institute, Charles University in Prague, Malostransk{\'e} n{\'a}m{\v e}st{\'\i} 25,
        118 00 Prague, Czech Republic.}
\address[cunidam]{Department of Applied Mathematics, Faculty of Mathematics and
	   	Physics,\\Charles University in Prague, Malostransk{\'e} n{\'a}m{\v e}st{\'\i} 25,
        118 00 Prague, Czech Republic.}
\address[rutgers]{Department of Computer Science, Rutgers University, Piscataway, NJ}

\begin{abstract}
We study the computational complexity of graph planarization via edge contraction. 
The problem \contr\ asks whether there exists a set $S$ of at most $k$ edges that when contracted
produces a planar graph. We work with a more general problem called \restcontr{P} in which $S$, in addition, is
required to satisfy a fixed MSOL formula $P(S,G)$. We give an FPT algorithm in time $\O(n^2 f(k))$ which
solves \restcontr{P}, where $n$ is number of vertices of the graph and $P(S,G)$ is (i) inclusion-closed and (ii) inert contraction-closed
(where inert edges are the edges non-incident to any inclusion-minimal solution $S$).

As a specific example, we can solve the $\ell$-subgraph contractibility problem in which the edges
of the set $S$ are required to form disjoint connected subgraphs of size at most $\ell$. This problem
can be solved in time $\O(n^2 f'(k,\ell))$ using the general algorithm. We also show that for $\ell \ge
2$ the problem is \cNP-complete.
\end{abstract}


\maketitle

\section{Introduction} \label{sec:introduction}

Graph visualization techniques are thoroughly studied. In many applications visual understanding
of the graph under consideration is important or required. It is commonly accepted that edge
crossings make a plane drawing of a graph less clear, and thus the goal is to avoid them, or reduce
their number. It is now well-known that one can decide fast whether crossings can be avoided at all,
as planarity testing is linear time decidable~\cite{HT}, while determining the minimum number of
crossings needed to draw a non-planar graph is NP-hard~\cite{GJ}. Several variants of planar
visualization of graphs have been considered and explored, including simultaneous embeddings,
book-embeddings, embeddings on surfaces of higher genus, etc.

Marx and Schlotter~\cite{marx} considered planarization of a graph by removing its vertices while
Kawarabayashi and Reed~\cite{kawarabayashi_reed} considered removing its edges.  Another
possible way to planarize a graph is by contracting some of its edges. If the number of contracted
edges is not limited, every connected graph can trivially be contracted into a single vertex, and
thus becomes planar. A graph is \emph{$k$-contractible} if the number of contracted edges is limited
by a number $k$. If $k$ is a part of the input, testing
$k$-contractibility is \cNP-complete~\cite{Asano}. Polynomial-time algorithms are known if one asks
about contraction to a particular fixed planar graph (so called $H$-contractibility); for nice
overviews see~\cite{Hof,bipartite}.  In this paper, we present a fixed-parameter tractable algorithm
for contractibility to planar graphs.

\begin{figure}[t!]
\centering
\includegraphics{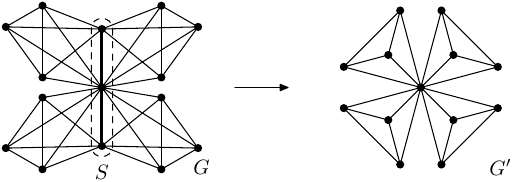}
\caption{An example of a 2-contractible graph.}
\label{fig:planarizing_set}
\end{figure}

\heading{Definitions and Notation.}
In this paper all graphs are simple, i.e., no multiple edges and no loops.
For a graph $G$, we denote by $V(G)$ its vertices and by $E(G)$ its edges (or simply $V$ and $E$,
when the graph is clear from the context). By $G \setminus H$, we denote the subgraph of $G$ induced
by $V(G) \setminus V(H)$.  We denote by $G \circ e$ the graph obtained by contracting an edge $e$ in
$G$.  For a set of edges $S$, we denote a graph created from $G$ by contracting all edges of $S$ by
$G \circ S$. We call $S \subseteq E$ \emph{a planarazing set} of $G$ if $G \circ S$ is a planar
graph; see Fig.~\ref{fig:planarizing_set}. We say that $G$ is \emph{$k$-contractible} if there
exists a planarizing set $S$ of size at most $k$.

MSOL formulas for graphs are logic formulas which contain predicates of equality, incidence and
containment, logic operators and quantifiers for vertices, edges, sets of vertices and edges;
see~\cite{courcelle_book}.  For instance, 3-colorability can be expressed by MSOL as existence of
three sets $V_1$, $V_2$ and $V_3$ of vertices, such that each vertex belongs to exactly one $V_i$,
and there are no edges with both endpoints in one $V_i$. 

\heading{Restricted Contractibility.}
We address the following more general problem. We want to find a planarizing set $S$ of size at most
$k$ that satisfies an additional restriction: a monadic second-order logic (MSOL) formula $P(S,G)$
fixed for the problem.\footnote{More precisely, we have different formulas $P_k(e_1,\dots,e_k)$ for
each $k$ where $S = \{e_1,\dots,e_k\}$. So the length of the formula may depend on $k$.}

\computationproblem
{\restcontr{P}}
{An undirected graph $G$ and an integer $k$.}
{Is there a planarizing set $S\subseteq E(G)$ of size at most $k$ satisfying $P(S,G)$ that when contracted
produces a planar graph?}

We want to construct an FPT algorithm for this problem with respect to the parameter $k$. This is
not possible for every MSOL formula $P(S,G)$. For some formulas, the problem is already \cNP-hard
even for $k=0$. For instance, let $P(S,G)$ be the formula: ``For $S = \emptyset$, is $G \circ S$
a 3-colorable graph?'' Then the problem \restcontr{P} is equivalent to testing 3-colorability of planar graphs
which is known to be \cNP-complete~\cite{3-color}.

In this paper, we describe an FPT algorithm which works as follows.  Either the graph is simple (of
small tree-width) and the problem can be solved in a brute-force way. Or we find a small part of the
graph which we can prove to be far from any inclusion-minimal planarizing set. We modify the graph
by contracting this small part, and repeat the process.  Therefore, we need to restrict ourselves to
MSOL formulas for which satisfiablity is not changed by this modification. 

An MSOL formula $P$ is \emph{inclusion-closed} if for every $S$ satisfying $P$ also every $S'
\subseteq S$ satisfies $P$. This property is necessary since the algorithm looks for
inclusion-minimal planarizing sets. A set $B$ of edges of $G$ is called \emph{inert} if it is not
incident with any edge of any inclusion-minimal planarizing set $S$; see Fig.~\ref{fig:inert_sets}a
for an example. A formula $P$ is called \emph{inert contraction-closed} if the following holds for
every inclusion-minimal planarizing set $S$ and every inert set $B$
\begin{displaymath}
P(S, G) \Longleftrightarrow P(S, G \circ B)
\end{displaymath}
Therefore the modification by contraction of inert edges does not change solvability of the problem.

\begin{figure}[t!]
\centering
\includegraphics{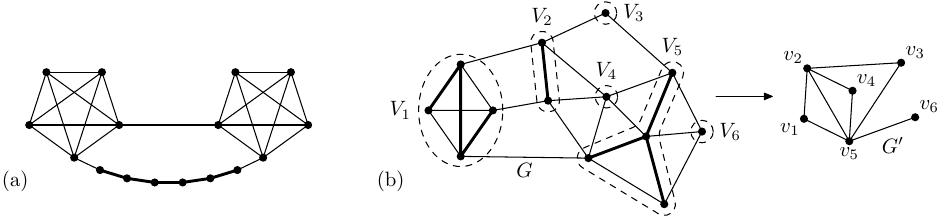}
\caption{(a) Every inclusion-minimal planarizing set $S$ contains one edge in each $K_5$. Therefore,
inert sets are subsets of the highlighted edges. (b) The connected components of a planarizing sets
$S$ act as clusers. After contracting $S$, each cluster corresponds to one vertex in $G'$. Two
vertices in $G'$ are adjacent if and only if there exists an edge in $G$ between the corresponding
clusters.}
\label{fig:inert_sets}
\end{figure}

\begin{theorem} \label{thm:fpt}
For every inclusion-closed and inert contraction-closed MSOL formula $P$, the problem \restcontr{P}\
is solvable in time $\O(n^2 f(k))$ where $n$ is the number of vertices of $G$ and $f$ is a computable function. 
\end{theorem}

Our algorithm uses an approach developed by Grohe~\cite{grohe} which shows that there is
a quadratic-time FPT algorithm for crossing number. The most significant difference is the proof of
Lemma~\ref{lem:flat}.  We cannot use the same approach as that of~\cite{grohe} because
$k$-contractible graphs do not have bounded genus~\cite{Gol} which is essential in~\cite{grohe}.
Further, our approach in the proof of Lemma~\ref{lem:flat} can be modified for the crossing number
which simplifies the argument of Grohe~\cite{grohe}; see Section~\ref{sec:grohe}.

For a trivial formula $P$ that is true for every set of edges, we get the $k$-contractibility
problem considered above.  We note that $k$-contractibility was independently proved to be solvable
in time $\O(n^{2+\eps}\bar{f}(k))$ for every $\eps > 0$ in a recent paper of Golovach et
al.~\cite{Gol}. The algorithm described here uses similar techniques but has a better time
complexity.

\heading{$\ell$-subgraph Contractibility.}
For different formulas $P$, we get problems different from $k$-contractibility having new specific
properties.  As one particular example, we work with a problem which we call $\ell$-subgraph
contractibility. A graph is called \emph {$\ell$-subgraph contractible} if and only if there exists
a planarizing set $S$ such that its edges form disjoint connected subgraphs with at most $\ell$
vertices. For instance, for $\ell=2$ the planarizing set $S$ is required to be a matching, see
Fig.~\ref{fig:matching_planarizing_set}.

\begin{figure}[b!]
\centering
\includegraphics{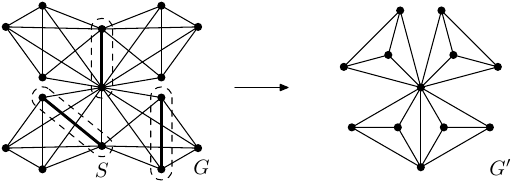}
\caption{For $\ell=2$, when we require that a planarizing set $S$ is a matching, a smallest
planarizing set contains three edges.}
\label{fig:matching_planarizing_set}
\end{figure}

\computationproblem
{\subcontr}
{An undirected graph $G$ and an integer $k$.}
{Is $G$ $\ell$-subgraph contractible by a set $S$ having at most $k$ edges?}

Contraction of a set $S$ can be interpreted as graph clustering, see Fig.~\ref{fig:inert_sets}b.
We want to find clusters such that the resulting cluster graph is planar.  For $\ell$-subgraph
contractibility, every cluster has to be of size at most $\ell$.  In comparison to
$k$-contractibility, the contracted edges have to be more equally distributed in $G$, and thus the
contractions do not change the graph too much.
	
From a graph drawing perspective this approach offers a drawing such that all crossings happen
in disjoint areas nearby the clusters and the rest of the meta-drawing is crossing-free.  Such a
meta-drawing resembles well the original graph and can be well grasped by a glance from the
distance. The local crossings get inspected by taking a magnifying glass for particular clusters.

If $\ell = 1$, the problem is solvable in linear time as it becomes just planarity testing.
For $\ell \ge 2$, we prove:

\begin{proposition} \label{prop:npc}
For $\ell \ge 2$, the problem \subcontr\ is \cNP-complete.
\end{proposition}

Since $\ell$-subgraph contractibility can be expressed using MSOL formulas, we get the following
corollary of Theorem~\ref{thm:fpt}.

\begin{corollary} \label{cor:subcontr}
For every fixed $\ell$, the problem \subcontr\ can be solved in time $\O(n^2 f'_\ell(k))$ where $n$ is
the number of vertices and $f_l'$ is a computable function.
\end{corollary}

\heading{Paper Layout.}
In Section~\ref{sec:fpt}, we describe our FPT algorithm for the \restcontr{P} problem. In
Section~\ref{sec:subcontr}, we deal with the \subcontr\ problem. Last, in Section~\ref{sec:grohe},
we show how to simplify the proof of Grohe~\cite{grohe}.

\section{Restricted Contractibility is Fixed-Parameter Tractable} \label{sec:fpt}

Let $P$ be a fixed inclusion-closed and inert contraction-closed MSOL formula.  In this section, we
show that the problem \restcontr{P}\ is fixed-parameter tractable with respect to the parameter $k$.
Namely, we describe an algorithm which solves \restcontr{P}\ in time $\O(n^2 \cdot f(k))$ for some function $f$.

The basic structure of our algorithm is based on the following idea invented by Grohe~\cite{grohe}.
If the graph has a small tree-width, we solve the problem by Courcelle's Theorem~\cite{courcelle}.
If the tree-width is large, we find an embedded large hexagonal grid and produce a smaller graph to
which we apply the procedure recursively.

\subsection{Definitions}

\begin{figure}[b!]
\centering
\includegraphics{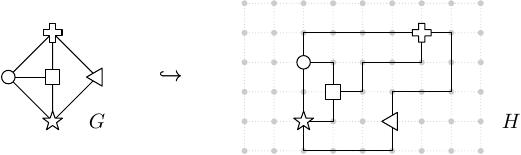}
\caption{An example of a topological embedding $G \embed H$.}
\label{fig:topological_embedding}
\end{figure}

We first introduce notation similar to that of Grohe's in~\cite{grohe}.

\heading{Topological Embeddings.} A topological embedding $h: G \embed H$ of $G$ into $H$ consists of
two mappings: $h_V: V(G) \to V(H)$ and $h_E: E(G) \to P(H)$, where $P(H)$ denotes set of all paths in $H$.
These mappings must satisfy the following properties:
\begin{packed_itemize}
\item The mapping $h_V$ is injective, distinct vertices of $G$ are mapped to distinct
vertices of $H$.
\item For distinct edges $e$ and $f$ of $G$, the paths $h_E(e)$ and $h_E(f)$ are distinct, do not
share internal vertices and share possibly at most one endpoint.
\item If $e = uv$ is an edge of $G$ then $h_V(u)$ and $h_V(v)$ are the endpoints of the path
$h_E(e)$. If $w$ is a vertex of $G$ different from $u$ and $v$ then path $h_E(e)$ does not contain
the vertex $h_V(w)$.
\end{packed_itemize}
For an example, see Fig.~\ref{fig:topological_embedding}. 

It is useful to notice that there exists a topological embedding $h: G \embed H$, if there exists a
subdivision of $G$ which is a subgraph of $H$. For a subgraph $G' \subseteq G$, denote by $h \rest G'$
the restriction of $h$ to $G'$. For simplicity, we use the term \emph{embeddings} instead of
topological embeddings.

\heading{Hexagonal grid.} We define recursively the hexagonal grid $H_r$ of radius $r$  (see
Fig.~\ref{fig:hexgrid}). The graph $H_1$ is a hexagon (the cycle of length six). The graph
$H_{r+1}$ is obtained from $H_r$ by adding $6r$ hexagonal faces around $H_r$ as indicated in
Fig.~\ref{fig:hexgrid}.

The nested \emph{principal cycles} $C^1, \dots, C^r$ are called the boundary cycles of
$H_1,\dots,H_r$ . From the inductive construction of $H_r$, $H_k$ is obtained from $H_{k-1}$ by
adding $C^k$ and connecting it to $C^{k-1}$. A \emph{principal subgrid} $H^s_r$ where $s \le r$
denotes the subgraph of $H_r$ isomorphic to $H_s$ and bounded by the principal cycle $C^s$ of $H_r$.

\begin{figure}[t!]
\centering
\includegraphics{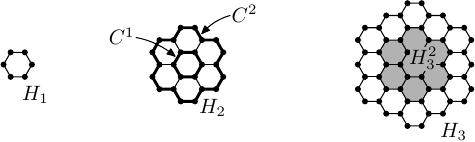}
\caption{Hexagonal grids $H_1$, $H_2$ and $H_3$. Inside $H_2$, the principal cycles $C^1$ and $C^2$ are
depicted in bold. Inside $H_3$, the principal subgrid $H^2_3$ is shown.}
\label{fig:hexgrid}
\end{figure}

\heading{Flat Topological Embeddings.}
Let $H$ be a subgraph of a graph $G$. An $H$-\emph{component} $C$ of $G$ is
\begin{packed_itemize}
\item either a connected component of $G \setminus H$ together with the edges connecting $C$ to $H$
and its incident vertices, or 
\item an edge $e = uv$ and the incident vertices $u$ and $v$ such that $u,v \in V(H)$ and $e \notin
E(H)$.
\end{packed_itemize}
The endpoints of edges of $C$ contained in $H$ are called the \emph{vertices of attachment} of $C$.
Figure~\ref{fig:components}a illustrates the notion of $H$-components.

\begin{figure}[b!]
\centering
\includegraphics{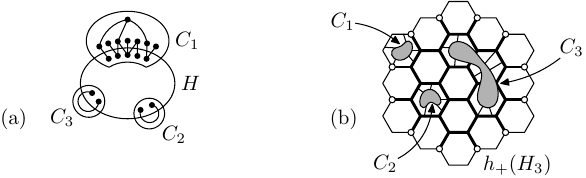}
\caption{(a) $H$-components $C_1$, $C_2$ and $C_3$ of $G$. (b) A subgraph $h_+(H_3)$ consisting of
$h(H_3)$ and three proper components $C_1$, $C_2$ and $C_3$ having attachments to inner vertices of
$h(H_3)$ (highlighted in bold). The embedding $h$ is not flat since $C_3$ obstructs
planarity.}
\label{fig:components}
\end{figure}

Let $G$ be a graph and let $h: H_r \embed G$ be an embedding of a hexagonal grid in $G$. A vertex $v \in h(H_r)$ is called \emph{inner} if $v \in h(H_r) \setminus h(C^r)$.
An $h(H_r)$-component $C$ is called \emph{proper} if $C$ has at least one vertex of attachment
in $h(H_r) \setminus h(C^r)$, namely, the component is attached to an inner vertex of the grid.
Let $h_+(H_r)$ denote the union of $h(H_r)$ with all proper $h(H_r)$-components. Notice that the proper
$h(H_r)$-components may be obstructions to the  planarity of $h_+(H_r)$. The embedding $h$ is called
a \emph{flat embedding} if $h_+(H_r)$ is a planar graph. For an example, see
Fig.~\ref{fig:components}b.

\heading{Tree-width.}
For a graph $G$, its tree-width is an integer $k$ which describes how ``similar'' is $G$ to a
tree~\cite{diestel}.  For our purposes, we use tree-width as a black box in our algorithm. The
following two properties of tree-width are crucial.

\begin{theorem}[Robertson and Seymour~\cite{robsey8},\ \,Boadlaender~\cite{bodla},\,\ and
Perkovi\`{c} and Reed~\cite{perko}]
\label{thm:robsey}
For every $s \ge 1$, there is $t \ge 1$ and a linear-time algorithm that, given a graph
$G$, either (correctly) recognizes that the tree-width of $G$ is at most $t$ or returns an embedding
$h: H_s \embed G$.
\end{theorem}

\begin{theorem}[Courcelle~\cite{courcelle}]
\label{thm:courcelle}
For every graph $G$ of tree-width at most $t$ and every MSOL formula $\varphi$, there exists an
algorithm that decides the formula $\varphi$ on $G$ in time $\O(n \cdot g(t,|\varphi|))$, where $n$ is the number of
vertices of $G$.
\end{theorem}

\subsection{The algorithm}

\heading{Overview.} The general outline of the algorithm is as follows. It proceeds
in two phases. The first phase deals with graphs of large tree-width and repeatedly modifies $G$
to produce a graph of small tree-width. In addition, we keep a set $F \subseteq E$ of forbidden
edges for contractions. Initially, $F$ is empty and during the modification some edges are added.
The second phase reduces \restcontr{P} to solving an MSOL formula which is done by Courcelle's
Theorem~\ref{thm:courcelle}.

\heading{Phase I.} We first prove the following lemma, which states that in an embedded large hexagonal grid
$H_s$ into $G$, we either find a flat hexagonal grid $H_r$ (smaller than $H_s$), 
or else $G$ is not $k$-contractible. This lemma represents the most significant difference from the paper of
Grohe~\cite{grohe}, as illustrated in Section~\ref{sec:grohe}.

\begin{figure}[b!]
\centering
\includegraphics{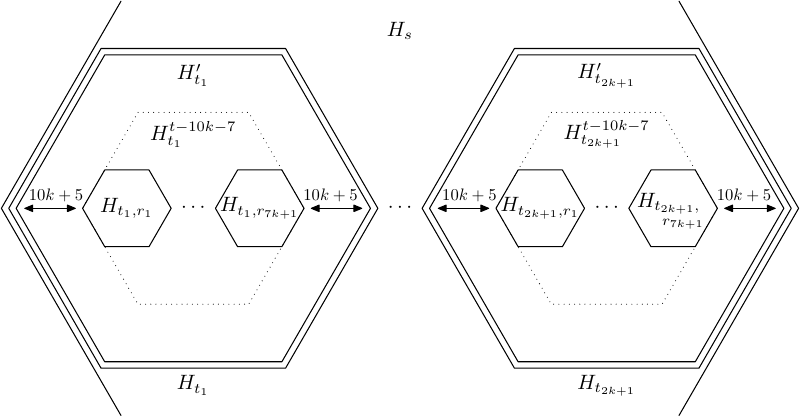}
\caption{The hierarchy of the hexagonal grids nested in $H_s$.}
\label{fig:grid_hierarchy}
\end{figure}

\begin{lemma}
\label{lem:flat}
Let $G$ be a $k$-contractible graph. For every $r \ge 1$, there exists $s \ge 1$ such that for
every embedding $h: H_s \embed G$ there is some subgrid $H_r \subseteq H_s$ such that
$h \rest H_r$ is a flat embedding.
\end{lemma}

\begin{proof}
For given $r$ and $k$, we fix $s$ and $t$ large enough as follows: We choose $s \approx 2kt$ so that
$H_s$ contains $2k+1$ \emph{disjoint subgrids} $H_{t_1},\dots,H_{t_{2k+1}}$ of radius $t$. Let
$H'_{t_i}$, formerly denoted by $H_{t_i}^{t-2}$, be the principal subgrid of $H_{t_i}$ obtained from
$H_{t_i}$ by removing the two outermost layers.  We choose $t \approx (7r+10)k$ so that each principal
subgrid $H^{t-10k-7}_{t_i}$ contains $7k+1$ \emph{disjoint subgrids}
$H_{t_i,r_1},\dots,H_{t_i,r_{7k+1}}$ of radius $r$.\footnote{Actually just $s \approx \sqrt {2k} t$
and $t \approx \sqrt{7k}r+10k$ would be sufficient.} In this way, we get the hierarchy of nested
subgrids as in Fig.~\ref{fig:grid_hierarchy}:
$$H_s \supsetneq H_{t_i} \supsetneq H'_{t_i} \supsetneq H_{t_i,r_j},\qquad
		\text{where } 1 \le i \le 2k+1,\quad 1 \le j \le 7k+1.$$
We next argue using the pigeon-hole principle that for some $H_{t_i,r_j}$ the embedding $h \rest
H_{t_i,r_j}$ is a flat embedding.

Since we are assuming that $G$ is $k$-contractible, we can fix a corresponding planarizing set $S$
and consider one subgrid $H_{t_i}$. Let a \emph{cell} be the $h$-image of a hexagon of $H'_{t_i}$. We
call an $h(H_{t_i})$-component \emph{bad} if it contains an edge from $S$. A cell is considered
\emph{bad} if it contains an edge of $S$ or if there is at least one bad $h(H_{t_i})$-component
attached to the cell. Since bad cells have some obstructions to planarity attached to them, we will
exhibit some grid $H_{t_i,r_j}$ such that its embedding $h(H_{t_i,r_j})$ avoids all bad cells.

To proceed, call an $h(H_{t_i})$-component $C$ \emph{large} if it has two vertices of attachment of
$C$ in $h(H'_{t_i})$ which do not belong to one cell of $H'_{t_i}$. As an example, in
Fig.~\ref{fig:components}b (when $H_3 = H'_{t_i}$) the component $C_3$ is large, but $C_1$ and $C_2$
are not. Large $h(H_{t_i})$-components posses the following useful properties which we prove
afterwards in a series of claims.

\begin{packed_enum}
\item If an $h(H_{t_i})$-component is large, then we can embed $K_{3,3}$ into $h_+(H_{t_i})$. This
implies that there must be some $H_{t_i}$ having no large $h(H_{t_i})$-component, otherwise the graph
would not be $k$-contractible.
\item On the other hand,  if a bad $h(H_{t_i})$-component is not large, it can produce at most
seven bad cells. This implies that $h(H_{t_i})$ must have a number of bad cells bounded by $7k$ and
therefore for some $j$, the embedding $h(H_{t_i,r_j})$ contains no bad cells and it is a flat embedding.
\end{packed_enum}

\begin{claim}
Let $C$ be a large $h(H_{t_i})$-component. Then we can embed $K_{3,3}$ into $h_+(H_{t_i})$ such that
$K_{3,3}^- := K_{3,3} - e$ is embedded into the grid $h(H_{t_i})$.
\end{claim}

\begin{claim_proof}
Instead of a tedious formal proof, we illustrate the main idea in Fig.~\ref{fig:bridge}. If $C$ is
large, it has two vertices $u$ and $v$ in $h(H'_{t_i})$ not contained in one cell. Thus there exists
a path $P$ going across the grid ``between'' $u$ and $v$. Using $C$ as a ``bridge'' from $u$ to $v$,
we can cross $P$ by another path across the grid. These two paths together with the two outer layers of
$h(H_{t_i})$ allow to embed $K_{3,3}$ into $h_+(H_{t_i})$ such that $K_{3,3}^-$ is embedded into
$h(H_{t_i})$.\claimqed

\begin{figure}[t!]
\centering
\includegraphics{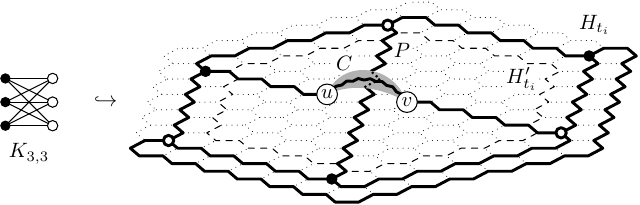}
\caption{The component $C$ acts as a bridge allowing two non-crossing paths across the grid. Thus we
can embed $K_{3,3}$ into $h(H_{t_i})$.}
\label{fig:bridge}
\end{figure}
\end{claim_proof}

\begin{claim}
There is some $H_{t_i}$ such that there is no large $h(H_{t_i})$-component.
\end{claim}

\begin{claim_proof}
According to the above claim, if there exists a large $h(H_{t_\ell})$-component, we can
embed $K_{3,3}$ into $h_+(H_{t_\ell})$. Since $G \circ S$ is a planar graph, this embedding of
$K_{3,3}$ has to be contracted by $S$. To contract it, there has to be an edge $e \in S$
incident with $h(K_{3,3}^-)$, otherwise $G \circ S$ still contains an embedding of $K_{3,3}$.
Therefore $e$ is incident with $h(H_{t_\ell})$. We know that $|S| \le k$ and each edge in $S$ is
incident with at most two grids $h(H_{t_\ell})$. Since we have $2k+1$ disjoint grids, there is
some $H_{t_i}$ such that no edge of $S$ is incident with $h(H_{t_i})$.\footnote{By a more refined
analysis, one can show that only $k+1$ disjoint grids suffice. The reason is that if an edge is
incident to two grids, it contracts neither of the embeddings $h(K_{3,3})$.} Therefore, there is
no large $h(H_{t_i})$-component.\claimqed
\end{claim_proof}

\begin{claim}
For $H_{t_i}$ having no large $h(H_{t_i})$-component, $h(H'_{t_i})$ contains at most
$7k$ bad cells.
\end{claim}

\begin{claim_proof}
Let $e \in S$. Each $h(H_{t_i})$-component $C$ is not large, so it is attached to at most seven
cells of $h(H'_{t_i})$. Therefore, if $e \in C$, we get at most seven bad cells.  If $e$ belongs to
a cell directly, we get two bad cells. Since $|S| \le k$, we get at most $7k$ bad cells.\claimqed
\end{claim_proof}

By the
pigeon-hole principle, there exists one of $h(H_{t_i,r_1}),\dots,h(H_{t_i,r_{7k+1}})$ containing no
bad cells, and we denote it by $H_{t_i,r_j}$. 

\begin{claim} \label{clm:flat}
For $H_{t_i}$ having no large $h(H_{t_i})$-component and $h(H_{t_i,r_j})$ having no bad cells, the
embedding $h \rest H_{t_i,r_j}$ is a flat embedding.
\end{claim}

\begin{claim_proof}
Since $h(H_{t_i,r_j})$ contains no bad cells, clearly it contains no edges of $S$. Let $C$ be a
proper $h(H_{t_i,r_j})$-component, it remains to show that $C \cap S = \emptyset$. Notice that all
above claims involve $h(H_{t_i})$-components, so we need to relate $C$ to them.  If $C$ is an edge,
let $C' = C$. If $C$ contains an edge with one endpoint in $h(H_{t_i,r_j})$ and another endpoint in
$h(H_{t_i}) \setminus h(H_{t_i,r_j})$, let $C'$ be this edge.  Otherwise let $C'$ be a component of $C
\setminus h(H_{t_i})$ together with the edges connecting $C'$ to $h(H_{t_i})$ and its incident
vertices, and assume that $C'$ is attached to an inner vertex $u$ of $h(H_{t_i,r_j})$; it always
exists. Observe that $C'$ is a $h(H_{t_i})$-component.

Since $h(H_{t_i,r_j})$ contains no bad cells, $C' \cap S = \emptyset$.  Also, $C'$ is not attached
to any vertex of $h(H'_{t_i} \setminus H_{t_i,r_j})$, otherwise it would be a large component.
It remains to show that $C'$ is not attached to any vertex of $h(H_{t_i} \setminus H'_{t_i})$ as
well. This concludes the proof since $C \cap h(H_{t_i} \setminus H_{t_i,r_j}) = \emptyset$, so $C 
= C'$ and $C \cap S = \emptyset$.

\begin{figure}[b!]
\centering
\includegraphics{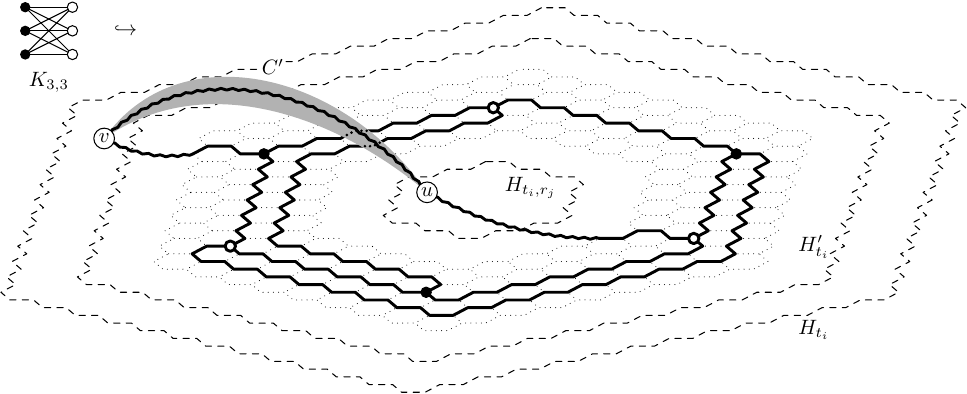}
\caption{When 5 consecutive layers of $h(H_{t_i} \setminus H_{t_i,r_j})$ (depicted with dots) avoid
the edges of $S$, we can use the component $C'$ to embed $K_{3,3}$.}
\label{fig:algorithm8_flat_embedding}
\end{figure}

Suppose that $C'$ is attached to $v \in h(H_{t_i} \setminus H'_{t_i})$. By our assumption, $u$ and
$v$ are at least $10k+5$ layers of the grid $h(H_{t_i})$ apart.  By the pigeon-hole principle, there
exist 5 consecutive layers $L = h(H_{t_i}^q \setminus H_{t_i}^{q-5})$ of the grid $h(H'_{t_i}
\setminus H_{t_i,r_j})$ such that no edge of $S$ is incident with them.  Using $L$, we can embed
$K_{3,3}$ into $h_+(H_{t_i}) \circ S$; see Fig.~\ref{fig:algorithm8_flat_embedding}. We embed
$K_{3,3}-e$ into the middle three layers of $L$. The remaining edge $e$ is embedded in the
outer/inner layers of $L$, together with a path in $h(H_{t_i}) \cup C' \setminus L$, using a path
from $u$ to $v$ through $C'$ as a bridge.  Since edges of $S$ are not incident with $L$, they are
not incident with the embedding of $K_{3,3}-e$, so $K_{3,3}$ remains in $G \circ S$, contradicting
that $S$ is a planarizing set.  \claimqed
\end{claim_proof}

Following the above claims, we get for $H_r = H_{t_i,r_j}$ a flat embedding $h \rest H_r$,
concluding the proof.\qed
\end{proof}

The next lemma shows that a small part of $h_+(H_r)$, where $H_r$ is the hexagonal grid which we got from 
Lemma~\ref{lem:flat}, is never contracted by a minimal set $S$. Let the \emph{core} $K$ of $h_+(H_r)$ 
denote the $h$-image of the central principal cycle $h(C^1)$ together with the $h(H_r)$-components 
attached only to $h(C^1)$.

\begin{lemma} \label{lem:core}
Let $G$ be a $k$-contractible graph and let $S$ be a minimal planarizing set of $G$. Let $h:
H_r \embed G$, where $r \ge 4k+3$, be a flat embedding and let $K$ be the core of $h_+(H_r)$. Then the
edges of $G$ incident with the vertices of $K$ do not belong to $S$.
\end{lemma}

\begin{proof}
We define a \emph{ring} $R_i$ as $h(H_r^{2i+1} \setminus H_r^{2i-1})$, i.e., it is the $h$-image of
the two consecutive principal cycles of $H_r$ together with the edges between them. An
\emph{extended ring} $R_i^+$ is the union of $R_i$ with all the $h(H_r)$-components having  all
vertices of attachment in $R_i$. See Fig.~\ref{fig:ring2}a. Consider $2k+1$ disjoint extended rings
$R_1^+,\dots,R_{2k+1}^+$ and notice that for some $i$, no edge of $S$ is incident with a vertex of
$R_i^+$. Denote $R := R_i$, $R^+ := R_i^+$ and the principal cycles of this ring as
$C_\outer := h(C^{2i+1})$ and $C_\inner := h(C^{2i})$.

\begin{figure}[t!]
\centering
\includegraphics{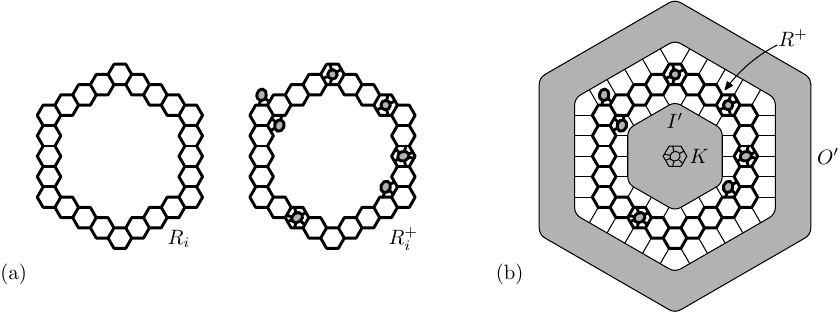}
\caption{(a) An example of a ring $R_i$ and an extended ring $R_i^+$. (b) The extended ring $R^+$
(in bold) splits $G' \setminus R$ into an outer part $O'$ and an inner part $I'$. The inner part is
connected and contains $H_r^{2i-1}$ and $K$.}
\label{fig:ring2}
\end{figure}

The graph $G' := G \circ S$ is planar with a planar embedding $\Delta'$ as follows; see
Fig.~\ref{fig:ring2}b.  The ring $R$ itself is a subdivision of a 3-connected graph, and therefore
it has a unique embedding up to the choice of an outer face. We choose a drawing having $C_\outer$
as the outer face of $\Delta' \rest R$. For the extended ring $R^+$, the embedding of the attached
components is not unique, but they are attached somehow to $R$.

Now, $G' \setminus R^+$ is split by $R^+$ into two parts: the \emph{inner part} $I'$ lying inside the face
bounded by $C_\inner$ and the \emph{outer part} $O'$ inside the outer face bounded by $C_\outer$.
Moreover, we can assume for $\Delta'$ that $I'$ is connected, containing $h(H_r^{2i-1}) \circ
S$. The reason is that a connected component of $I'$, not containing the grid, is a separate component
of $G'$, so we choose $\Delta'$ such that it is drawn into the outer face.

Since contraction does not change connectivity in $G$ and all contractions avoid $R^+$, the subgraph
$G \setminus R^+$ is separated into parts $I$ (containing $h(H_r^{2i-1})$ and especially $K$) and
$O$ (containing the rest) such that $I \circ S = I'$ and $O \circ S = O'$ with no edges between $I$
and $O$. We show next that the minimality of $S$ forces that $I \cap S = \emptyset$.

We take the embedding $\Delta'$, remove $\Delta' \rest I'$ and replace it with some embedding of
$I$.  Since $h(H_r)$ is a flat embedding, the graph $h_+(H_r)$ is planar, so in particular the
subgraph induced by $R^+ \cup I$ is planar. We consider one of its planar embeddings $\Delta$ having
$I$ embedded into the inner face of $R^+$. Since $\Delta \rest I$ has the same orientation of edges
to $R$ as $\Delta' \rest I'$, it is possible to replace the embedding of $I'$ in $\Delta'$ by
$\Delta \rest I$. This means that it is not necessary to contract the edges of $I$ and therefore $I
\cap S = \emptyset$. The statement follows since $K$ and its incident edges belong to $I$.\qed
\end{proof}

We note that the algorithm just contracts $K$ and not the whole $I$ since we do not know which
extended ring $R_i^+$ avoids edges of $S$.

\begin{figure}[t!]
\centering
\includegraphics{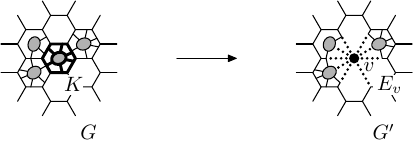}
\caption{The graph is modified by contraction of the core $K$ and adding its incident edges $E_v$
to $F$.}
\label{fig:core}
\end{figure}

Recall that $F$ is the set of forbidden edges to contract. If the graph $G$ is $k$-contractible by a
planarizing set $S$ such that $S \cap F = \emptyset$, we say that $G$ is
\emph{$(k,F)$-contractible}. As we proceed with our algorithm, we use Lemma~\ref{lem:core} to modify the
input $G$ and $F$ to a smaller graph $G'$ and an extended set $F'$ as follows.  The core $K$ is
contracted into a single vertex $v$, so $G' := G \circ K$.  Let $E_v$ be the set of edges incident
with $v$. We add them to $F$, so $F' := F \cup E_v$.  Figure~\ref{fig:core} depicts this
modification.

\begin{lemma} \label{lem:modify}
The graph $G$ is $(k,F)$-contractible if and only if the graph $G'$ is $(k,F')$-contractible.
\end{lemma}

\begin{proof}
According to Lemma~\ref{lem:core}, a minimal planarizing set $S$ for $G$ avoiding $F$
does not contain any edges of $K$ and also the edges of $G$ from which $F'$ arises in $G'$. 
Therefore, it is also a planarizing set of $G'$ avoiding $F'$.

On the other hand, assume that $G'$ has a planarizing set $S$ disjoint from $F'$ of size at most $k$.
We want to show that $S$ is also a planarizing set for $G$. We consider an embedding of $G' \circ S$
and we replace the vertex created by contracting the core $K$ by an embedding of $K$ in a manner
completely analogous to the one described in the proof of Lemma~\ref{lem:core}. We get a planar
embedding of $G \circ S$.\qed
\end{proof}

\heading{Phase II.}
If the graph $G$ is $k$-contractible by a planarizing set $S$ such that $S \cap F = \emptyset$ and
$P(S,G)$ is satisfied, we say that $G$ is \emph{$(k,F)$-contractible with respect to $P$}. As we
process our algorithm, we use Lemma~\ref{lem:core} to modify We show next that when the tree-width
of $G$ is small, we can solve $(k,F)$-contractibility with respect to $P$ using Courcelle's
theorem~\cite{courcelle}. To this effect all we need to show is that it is possible to express
$(k,F)$-contractibility in the monadic second-order logic (MSOL).

\begin{lemma} \label{lem:minor}
For a fixed graph $H$, there exists an MSOL formula $\mu_H(S,G)$ which is satisfied if and only if
$G' := G \circ S$ contains $H$ as a minor.
\end{lemma}

\begin{proof}
We modify a well-known formula $\widetilde\mu_H(G)$ for testing whether $H$ is a minor of $G$. For $|H|
= \ell$, the formula $\widetilde\mu_H(G)$ tests whether there exist disjoint sets of vertices
$V_1,\dots,V_\ell$ (representing the sets of vertices contracted to vertices of $H$) such that
\begin{packed_itemize}
\item for every $v_iv_j \in E(H)$ there exists an edge between $V_i$ and $V_j$ in $G$, and
\item each set $V_i$ is connected in $G$.
\end{packed_itemize}

Let $S = \{e_1,\dots,e_k\}$ and $e_j = x_jy_j$. To test whether $H$ is minor in $G'$, we
require for every $j \in \{1,...,k\}$ that each $V_i$ either contains both endpoints of $e_j$, or none of them. Formally,
$$\mu_H(S,G) = \widetilde\mu_H(G) \wedge \!\!\!
\bigwedge_{\substack{1 \le i \le \ell \\ 1 \le j \le k}} \!\!(x_j \in V_i \iff y_j \in V_i).$$
\vskip -2.5em
\qed
\end{proof}

\begin{lemma} \label{lem:msol}
There exists a formula $\varphi_k(F,G)$ which is satisfiable if and only if
$G$ is $(k,F)$-contractible with respect to the MSOL formula $P$.
\end{lemma}

\begin{proof}
This formula is defined as follows:
$$\varphi_k(F,G) := \exists S \subseteq E(G) : |S| \le k \wedge (S \cap F = \emptyset) \wedge
P(S,G)\,\wedge \neg\mu_{K_5}(S,G) \wedge \neg\mu_{K_{3,3}}(S,G).$$
\vskip -3em
\qed
\end{proof}

\begin{algorithm}[t!]
\caption{\restcontr{P}} \label{alg:fpt}
\begin{algorithmic}[1]
\REQUIRE A graph $G$ and an inclusion-closed and inert contraction-closed formula $P(S,G)$. 
\ENSURE A planarizing set $S$ of size at most $k$ satisfying $P(S,G)$ if it exists.
\medskip

\STATE Initialize the set of forbidden edges $F := \emptyset$.
\STATE Depending on $k$, choose suitable $s \ge 1$ and $t \ge 1$ for Theorem~\ref{thm:robsey}.
\medskip
\WHILE{the tree-width of $G$ is larger than $t$}
	\STATE Find an embedding $h : H_s \embed G$ using Theorem~\ref{thm:robsey}.
	\STATE Find a subgrid $H_r$ such that $h \rest H_r$ is a flat embedding as described in Lemma~\ref{lem:flat}.
	\STATE Modify the graph as $G := G \circ K$ and $F := F \cup E_v$.
\ENDWHILE
\medskip
\RETURN a planarizing set $S$ satisfying $\varphi_k(F,G)$ using Theorem~\ref{thm:courcelle} if it exists.
\end{algorithmic}
\end{algorithm}

\heading{Putting all the Pieces Together.}  We finish this section with a proof of the announced
Theorem~\ref{thm:fpt}, stating that \restcontr{P}\ for an inclusion-closed and inert contraction-closed 
MSOL formula $P(S,G)$ is solvable in time $\O(n^2 \cdot f(k))$ for some function $f$.

\begin{proof}[Theorem~\ref{thm:fpt}]
See Algorithm~\ref{alg:fpt} for a pseudocode; Phase I corresponds to steps 3 to 7, Phase II corresponds to step 8. Depending on $k$, we choose a suitable value for $s$ so we can apply
Lemmas~\ref{lem:flat} and~\ref{lem:core}. By Theorem~\ref{thm:robsey}, there is a corresponding
value of $t$.

We repeat Phase I till the tree-width of $G$ becomes at most $t$. Every iteration of Phase I first finds embedding $h$ of a large
hexagonal grid $H_s$, by Theorem~\ref{thm:robsey} in linear time.  Using Lemma~\ref{lem:flat}, there
exists a subgrid $H_r$ such that $h(H_r)$ is flat.  Moreover, we can find such $H_r$ in time
$\O(k^2n)$ by testing planarity for all $h_+(H_{t_i,r_j})$. We contract the kernel $K$ and we modify
the graph $G$ and the set of forbidden edges $F$.  Lemmas~\ref{lem:core} and \ref{lem:modify} show
that this modification does not change the solvability of the problem.  After each modification, we get
a smaller graph $G$. Therefore we need to repeat this at most $\O(n)$ times, so the total running
time of Phase I is $\O(n^2 \cdot p(k))$ for some function $p$.

Let $G$ denote the original graph and let $G'$ denote the modified graph and let $F'$ denote the set of forbidden edges created by Phase I.  Phase
II uses Theorem~\ref{thm:courcelle} to solve the MSOL formula $\varphi_k(F',G')$ in time $\O(n \cdot
q(k))$.  By Lemmas~\ref{lem:flat}, \ref{lem:core} and \ref{lem:modify}, the modified graph $G'$ is
$(k,F')$ contractible if and only if the original graph is $k$-contractible. It remains to show that
none of the modifications changes the satisfiability of $P(S,G)$. Since $P(S,G)$ is
inclusion-closed, we can concentrate on inclusion-minimal planarizing sets in $G$ and $G'$. Further
by Lemma~\ref{lem:core}, each modification contracts some inert edges which are not incident with
any inclusion-minimal planarazing set $S$. Since $P(S,G)$ is inert contraction-closed, none of these modifications changes the solvability of $P(S,G)$.  So testing $\varphi_k(F',G')$ correctly tests
whether $G$ is $k$-contractible with respect to $P(S,G)$.

The overall complexity of the algorithm is $\O(n^2 \cdot f(k))$ for some function $f$.\qed
\end{proof}

\section{$\ell$-subgraph Contractibility} \label{sec:subcontr}

We first establish Corollary~\ref{cor:subcontr} which states that for a fixed $\ell$, testing
$\ell$-subgraph contractibility can be done in time $\O(n^2 \cdot f_\ell'(k))$ for some function
$f'_\ell$. It follows from Theorem~\ref{thm:fpt} and the fact that $\ell$-subgraph contractibility is
expressible using MSOL.

\begin{proof}[Corollary~\ref{cor:subcontr}]
We just describe in words how to construct the MSOL formula $P$ and we check that $P$ is inclusion-closed
and inert contraction-closed. The length of the formula may depend on $k$ and $\ell$.  The formula
$P$ tests whether there exists a decomposition of the edges in $S$ into pairwise disjoint sets
$E_1,\dots,E_k$ satisfying the following properties. Let $V_1,\dots,V_k$ denote the corresponding
sets of vertices incident with $E_1,\dots,E_k$, respectively. The formula $P$ is satisfied if and
only if $|V_i| \le \ell$ for each $i$ and the sets $V_1,\dots,V_k$ are pairwise disjoint. This solves
$\ell$-subgraph contractibility and is an inclusion-closed MSOL formula.

It remains to show that $P$ is inert contraction-closed. Assume that $P(S,G)$ is satisfiable and let
$S$ be an inclusion-minimal planarizing set satisfying $P(S,G)$. For every inert set $B$, by
definition, no edge of $B$ is incident with an edge of $S$. Therefore $S$ is a planarizing set of $G
\circ B$ consisting of $\ell$-subgraphs, and $P(S,G \circ B)$ is satisfiable. For the other
implication, if $S$ is an inclusion-minimal planarizing set satisfying $P(S,G \circ B)$, then it
also satisfies $P(S,G)$. The reason is that $B$ is inert, so no edge of $S$ is incident with an edge
of $B$.  Therefore, $S$ is a planarizing set of $G$ consisting of $\ell$-subgraphs.\qed
\end{proof}

\heading{Matching Contractibility.}
In the rest of this section, we establish \cNP-completeness of \subcontr. To do so, we first
introduce a new problem called \emph{matching contractibility}. The graph $G$ is \emph{$F$-matching
contractible} with respect to a set of edges $F$ if there exists a planarizing set $S$ which forms a
matching in $G$ and $S \cap F = \emptyset$.

\computationproblem
{\matchcontr}
{An undirected graph $G$ and a set of forbidden edges $F \subseteq E$.}
{Is $G$ an $F$-matching contractible graph?}

First, we show that $\ell$-subgraph contractibility can solve matching contractibility. 

\begin{lemma} \label{lem:subsolvesmatch}
Matching contractibility is polynomial-time reducible to $\ell$-subgraph contractibility, for any
fixed $\ell$.
\end{lemma}

\begin{proof}
For an input $G$ and $F$, we produce a graph $G'$ which is $\ell$-subgraph contractible if and only
if $G$ is $F$-matching contractible. We replace the edges of $G$ by paths:
\begin{packed_itemize}
\item if $e \in F$, then we replace it by a path of length $\ell$, and
\item if $e \notin F$, then we replace it by a path of length $\ell-1$.
\end{packed_itemize}
Also, we put $k = |E(G')|$ so only $\ell$-subgraphs restrict a planarizing set.

If a planarizing set $S$ is a matching in $G$ avoiding $F$, then we can contract the corresponding
paths in $G'$ by $\ell$-subgraphs. On the other hand, let $S'$ be a planarizing set of $G'$
consisting of $\ell$-subgraphs. First, we ignore each $\ell$-subgraph which does not contract one of
the paths of length $\ell - 1$, since its contraction preserves the topological structure of the graph.
If a path in $G'$ corresponding to $e \in E(G)$ is contracted, it has to be contracted by a single
$\ell$-subgraph.  In such a case, $e \notin F$, otherwise the path is too long. Also, the contracted
paths have to form a matching since the $\ell$-subgraphs cannot share the end-vertices belonging to
$G$. So the planarazing set $S'$ of $G'$ gives a planarizing set $S$ of $G$ which is a matching and
which avoids $F$.\qed
\end{proof}

\heading{Overview of the Reduction.}
To show \cNP-hardness of \matchcontr, we present a reduction from \clpsat.  An instance $I$ of
\clpsat\ is a Boolean formula in CNF such that each variable occurs in exactly three clauses, once
negated and twice positive, each clause contains two or three literals and the incidence graph of
$I$ is planar.  Fellows et al.~\cite{fellows} show that this problem is \cNP-complete.

Given a formula $I$, we construct a graph $G_I$ with a set $F_I$ of forbidden edges such that $G_I$
is $F_I$-matching contractible if and only if $I$ is satisfiable. The construction has 
a variable gadget $G_x$ for each variable $x$, and a clause gadget $H_c$ for each clause $c$.
All variable gadgets $G_x$ are isomorphic and we have two types of clause gadgets
$H_c$, depending on the size of $c$.  These gadgets consist of several copies of the graph $K_5$
with most of the edges in $F_I$. In Fig.~\ref{fig:gadgets}, the edges not contained in $F_I$ are
represented by bold lines. Each variable gadget contains three pendant edges that are identified
with certain edges of the clause gadgets, thus connecting the variable and clause gadgets.

\begin{figure}[t!]
\centering
\includegraphics{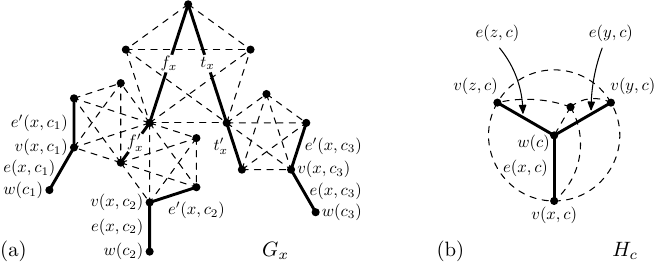}
\caption{The bold edges can be contracted and the dashed edges are forbidden edges from $F_I$. (a)
The variable gadget $G_x$ where the three outgoing edges are shared with clause gadgets $H_{c_1}$,
$H_{c_2}$ and $H_{c_3}$. (b) The clause gadget $H_c$. The edge $e(z,c)$ may also be in $F_I$ if the
clause contains only two variables $x$ and $y$. The two or three contractible edges are shared with
variable gadgets $G_x$, $G_y$ and possibly $G_z$.}
\label{fig:gadgets}
\end{figure}

\heading{Variable Gadget.}
Let $x$ be  a variable which occurs positively in clauses $c_1$ and $c_2$, and negatively in a clause
$c_3$. The corresponding \emph{variable gadget $G_x$} is depicted in Fig.~\ref{fig:gadgets}a. It consists
of four copies of $K_5$, each having all but two edges in $F_I$. Three of the copies of $K_5$ have
pendant edges attached, denoted by
$$e(x,c_i) = v(x,c_i)w(c_i),\qquad i \in \{1,2,3\};$$
refer to Fig.~\ref{fig:gadgets}a. These edges also belong to the clause gadgets $H_{c_1}$, $H_{c_2}$
and $H_{c_3}$. All other vertices and edges are private to the variable gadget $G_x$.

The main idea behind the variable gadget is that exactly one of the edges $t_x$ and $f_x$ is
contracted. This encodes the assignment of the variable~$x$ as follows: the edge $t_x$ is contracted
for true and $f_x$ for false. The edges $e(x,c_1)$ and $e(x,c_2)$ shared with the clause gadgets
$H_{c_1}$ and $H_{c_2}$, can be contracted only if $t_x$ is contracted, and $e(x,c_3)$ shared
with $H_{c_3}$ can be contracted if and only if $f_x$ is contracted.

\heading{Clause Gadget.}
Let $c$ be a clause containing variables $x$, $y$ and possibly $z$. The \emph{clause gadget $H_c$}
is the graph $K_5$ with all but 2 or 3 edges in $F_I$. The edges that are not forbidden to contract
share a common vertex $w(c)$ and they are the edges $e(x,c)$, $e(y,c)$ and possibly $e(z,c)$ shared
with the variable gadgets $G_x$, $G_y$ and possibly $G_z$; see Fig.~\ref{fig:gadgets}b.

To make the clause gadget planar, we need to contract exactly one of the edges $e(x,c)$, $e(y,c)$ and
possibly $e(z,c)$. This is possible only if the clause is satisfied by the corresponding variable
evaluated as true in this clause.

\begin{lemma} \label{lem:reduction}
The graph $G_I$ is $F_I$-matching contractible if and only if $I$ is satisfiable.
\end{lemma}

\begin{proof}
\noindent$\Longrightarrow$: Suppose first that $G_I$ is $F_I$-matching contractible, and let $S
\subseteq E(G_I) \setminus F_I $ be a matching planarizing set.  Using $S$, we construct a satisfying assignment of
$I$. Consider a variable $x$. In $G_x$, each copy of $K_5$ needs to have at least one edge contracted by
$S$.

Exactly one of $t_x$ and $f_x$ is in $S$. If $t_x \in S$, then $t'_x$ cannot be in $S$ (note that $S$
is a matching), hence $e'(x,c_3) \in S$, and $e(x,c_3)$ cannot be in $S$. On the other hand, if
$t_x \notin S$, necessarily $f_x \in S$, and by a similar sequence of arguments, none of $e(x,c_1)$
and $e(x,c_2)$ is in $S$.

We define a truth assignment for the variables of $I$ so that $x$ is true if and only if $t_x \in S$. It
follows that if $x$ appears as a false literal in a clause $c$, then the edge $e(x,c)$ is not in $S$.
Since $S$ contains exactly one edge of $H_c$, in each clause gadget at least one literal
must be evaluated to true. Thus $I$ is satisfiable.
\medskip

\noindent$\Longleftarrow$: Suppose that $I$ is satisfiable and fix a satisfying truth assignment
$\phi$. We set
$$
S=\{t_x,f'_x,e(x,c_1),e(x,c_2),e'(x,c_3) \mid \phi(x)=\text{true}\} \cup
\{f_x,t'_x,e'(x,c_1),e'(x,c_2),e(x,c_3) \mid \phi(x)=\text{false}\}.
$$
The edges of S contained in variable gadgets form a matching. Each clause gadget contains at least one edge
of $S$. But if a clause, say $c$, contains more than one literal evaluated as true, then its clause
gadget $H_c$ contains in $S$ more edges with the common vertex $w(c)$. In such a case, we perform a
\emph{pruning operation} on $H_c$, i.e, remove all edges from $S \cap E(H_c)$ but one. The resulting
set $S'$ is a matching such that each $K_5$ in $G_I$ contains exactly one edge of $S'$.

It only remains to argue that this set $S'$ is a planarizing set. The graph $G' = G \circ S'$
consists of copies of $K_4$ glued together by vertices or edges. Each copy is attached to other
copies by at most three vertices. Since $K_4$ itself has a non-crossing drawing in the plane such
that three of its vertices lie on the boundary of the outer face, and these vertices can be chosen
arbitrarily as well as their cyclic order on the outer face, the drawings of the variable and clause
gadgets can be combined together along a planar drawing of the incidence graph of $I$.
Therefore, $G'$ is planar.\qed
\end{proof}

This shows that \subcontr\ is \cNP-complete for every $\ell \ge 2$:

\begin{proof}[Proposition~\ref{prop:npc}]
Clearly, \subcontr\ belongs to \cNP. By Lemma~\ref{lem:reduction} and~\cite{fellows}, the problem
\matchcontr\ is \cNP-hard. Lemma~\ref{lem:subsolvesmatch} implies that \subcontr\ is \cNP-hard.\qed
\end{proof}

We note that the problem \subcontr\ remains \cNP-complete when generalized to surfaces of a fixed genus $g$
(instead of planar graphs). Consider a graph $H_g$ such that for every embedding of $H_g$ into the
surface, each face is homeomorphic to the disk. We modify our reduction by taking $G_I \cup H_g$ as
the graph and by adding all the edges of $H_g$ into $F$. For each surface, there exists such a graph
$H_g$ (see triangulated surfaces in Mohar and Thomassen~\cite{mohar}).

\section{Simplifying Grohe's Approach} \label{sec:grohe}

Grohe~\cite{grohe} gives an FPT algorithm for computing the crossing number $k$ of a graph, in time
$\O(n^2 f(k))$ for some function $f$. Our FPT algorithm is based on his approach. On the other hand,
we can simplify his argument in a similar manner as in the proof of Lemma~\ref{lem:flat}. We
describe this simplification here.

Grohe uses Thomassen's Theorem~\cite{thomassen} which states the following:

\begin{theorem}[Thomassen~\cite{thomassen}]
Let $G$ be a graph of genus at most $k$. For every $r \ge 1$, there is $s \ge 1$ such that for every
topological embedding $h: H_s \to G$, there exists a subgrid $H_r \subseteq H_s$ such that the
restriction $h \rest H_r$ of $h$ is flat.
\end{theorem}

\noindent This result can be used since for graphs the genus is upper bounded by the crossing number, so it
works as Lemma~\ref{lem:flat} in our FPT algorithm. With our simplification, we can completely avoid
the notion of genus which is independent of the notion of crossing number.  Therefore, we do not
need to consider the more complicated theory of graphs on surfaces. Also, our proof is quite short
and elementary.

Consider a plane drawing $D$ of a graph $G$.  The \emph{facial distance} $d(x,y)$ of vertices $x,y$
in $G$ is the minimal number of intersections of $D$ and a simple curve $\calC$ connecting $x$ and
$y$, where $\calC$ avoids the vertices of $G$.  Let $\crossnumber(G)$ denote the crossing number of
a graph $G$.  We use Riskin's Theorem~\cite{riskin}:

\begin{theorem}[Riskin~\cite{riskin}] \label{thm:riskin}
If $G$ is a 3-connected cubic planar graph, then
$$\crossnumber(G + xy) = d(x,y),$$
where $G + xy$ is the graph $G$ with the added edge $xy$. 
\end{theorem}

\begin{lemma}
Let $G$ be a graph with $\crossnumber(G) \le k$. For every $r \ge 1$, there is $s \ge 1$ such that for every
topological embedding $h: H_s \to G$, there exists a subgrid $H_r \subseteq H_s$ such that the
restriction $h \rest H_r$ of $h$ is flat.
\end{lemma}

\begin{proof}
Similarly as in Lemma~\ref{lem:flat}, we define a hierarchy of grids as in
Fig.~\ref{fig:grohe_grid_hierarchy}. We choose $s \approx
(k+r)\cdot k$ so that $H_s$ contains $k+1$ \emph{disjoint subgrids} $H_{t_1},\dots,H_{t_{k+1}}$ of radius
$k+r$. Let $H'_{t_i}$ denote the principal subgrid $H_{t_i}^{r}$ of $H_{t_i}$. In this way, we get the
hierarchy:
$$H_s \supsetneq H_{t_i} \supsetneq H'_{t_i},\qquad \text{where } 1 \le i \le k+1.$$
We want to show using the pigeon-hole principle that for some $H'_{t_i}$ the embedding $h \rest
H'_{t_i}$ is a flat embedding.

\begin{figure}[t!]
\centering
\includegraphics{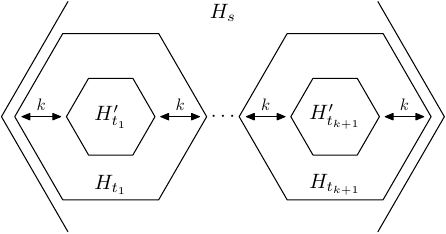}
\caption{The hierarchy of the hexagonal grids nested in $H_s$.}
\label{fig:grohe_grid_hierarchy}
\end{figure}

\begin{claim} \label{clm:grohe_pigeonhole}
Let $C$ be an $h(H_s)$-component which has a vertex of attachment in $h(H'_{t_i}) \setminus h(C_{t_i}^r)$, where $C_{t_i}^r$ is the $r$-th principal cycle of $H_{t_i}$.
Then $C$ has no attachment vertices in $h(H_s) \setminus h(H_{t_i})$.
\end{claim}

\begin{claim_proof}
For contradiction, let $x$ be an inner vertex of $h(H'_{t_i})$ which is a vertex of attachment of
$C$, and let $y$ be a vertex of attachment of $C$ in $h(H_s) \setminus h(H_{t_i})$.
Then there exists a path $P$ from $x$ to $y$ with no internal vertices in $h(H_s)$. Consider
$h(H_s)$ together with $P$ and apply Riskin's Theorem~\ref{thm:riskin}. (We note that subdividing
a graph preserves the crossing number.) Since the face distance $d(x,y)$ is at least $k+1$, we get
that $\crossnumber(G) > k$ which is a contradiction.\claimqed
\end{claim_proof}

Let $G_i$ be the graphs $h(H_{t_i})$ together with $h(H_{t_i})$-components with
vertices of attachment in $h(H'_{t_i}) \setminus h(C_{t_i}^r)$ (which are also $h(H_s)$-components, by
the previous claim).

\begin{claim} \label{clm:grohe_flat}
Let $G_i$ be the graph defined above. If $G_i$ is planar, then the embedding $h \rest H'_{t_i}$ is
flat.
\end{claim}

\begin{claim_proof}
Let $C$ be a proper $h(H'_{t_i})$-component. Exactly as in the proof of Claim~\ref{clm:flat}, we
construct from $C$ a proper $h(H_{t_i})$-component $C'$. Since $C'$ has a vertex of attachment in
$h(H'_{t_i}) \setminus h(C_{t_i}^r)$, it belongs to $G_i$.

The graph $h(H_{t_i})$ is a subdivision of a 3-connected planar graph, so it has a unique embedding
into the plane having $h(C_{t_i}^{r+k})$ as the outerface. Since $G_i$ is planar, the component $C'$
has to be embedded into a face bounded by a cell of $h(H'_{t_i})$ (which is the $h$-image of a
hexagon of $H'_{t_i}$). Therefore, $C'$ has no vertices of attachment in $h(H_{t_i}) \setminus
h(H'_{t_i})$ and $C = C'$. Thus, $h_+(H'_{t_i})$ is constructed from $G_i$ by removing $h(H_{t_i})
\setminus h(H'_{t_i})$, so it is planar.\claimqed
\end{claim_proof}

By Claim~\ref{clm:grohe_pigeonhole}, graphs $G_1,\dots,G_{k+1}$ are pairwise disjoint. 
Since $\crossnumber(G) \le k$, by the pigeon-hole principle, some $G_i$ is planar.
By Claim~\ref{clm:grohe_flat}, the embedding $h \rest H'_{t_i}$ is flat.\qed 
\end{proof}

\section{Open Problems} \label{sec:conclusions}

We conclude this paper with several open problems.

\begin{problem}
Let $P(S,G)$ be an inclusion-closed and inert contraction-closed MSOL formula.
Can the problem \restcontr{P} be solved in time $\O(n \cdot \widetilde f(k))$?
\end{problem}

\begin{problem}
Consider the generalization of \restcontr{P} which asks whether there exists a set $S \subseteq
E(G)$ such that $|S| \le k$ and $G \circ S$ is a graph of genus at most $g$. Is this problem
fixed-parameter tractable with the parameter $k$?
\end{problem}

\noindent In generalizing our approach, the main difficulty lies in Lemma~\ref{lem:flat}.

Asano and Hirata~\cite{Asano} proved that for an MSOL formula $P(S,G)$ which is always satisfied
\restcontr{P} is \cNP-complete.  In Section~\ref{sec:subcontr}, we show this for one particular
MSOL formula $P(S,G)$ which test whether $S$ consists of $\ell$-subgraphs. The very natural question
is for which formulas $P(S,G)$ it is \cNP-complete. Clearly, it is not for every formula $P(S,G)$.
For instance, if the formula $P(S,G)$ cannot be satisfied at all, the problem \restcontr{P} can be
solved by outputting ``no'' and clearly belongs to \cP.

\begin{problem}
For which MSOL formulas $P(S,G)$ is the problem \restcontr{P} \cNP-complete? 
\end{problem}

Theorem~\ref{thm:fpt} shows that for every inclusion-closed and inert contraction-closed formula
$P(S,G)$, the problem \restcontr{P} can be solved in FPT time with respect to the parameter $k$.
Can this be strengthened?

\begin{problem}
For which MSOL formulas $P(S,G)$ is the problem \restcontr{P} solvable in FPT time with respect to
the parameter $k$?
\end{problem}

\heading{Acknowledgement.} We would like to thank anonymous reviewers for comments which greatly
helped in improving this paper and fixing gaps in our arguments.

\bibliographystyle{elsarticle-num}
\bibliography{planar_contractions_journal}

\end{document}